\documentclass[11pt, final]{article}
\usepackage{amsthm,amsmath,amssymb,booktabs,xcolor,graphicx}
\usepackage{bbm}
\usepackage{listings}
\usepackage{xcolor}
\usepackage{appendix}
\usepackage{enumerate}
\usepackage[shortlabels]{enumitem}
\usepackage{verbatim}

\usepackage[margin=1in]{geometry}
\usepackage[utf8]{inputenc}
\usepackage[T1]{fontenc}

\usepackage[textsize=small,backgroundcolor=orange!20]{todonotes}
\usepackage[colorlinks=true, allcolors=blue]{hyperref}
\usepackage{url}
\usepackage{etoolbox}
\usepackage{appendix}
\usepackage[nameinlink, noabbrev,capitalize]{cleveref}
\crefname{equation}{}{} 
\AtBeginEnvironment{appendices}{\crefalias{section}{appendix}} 

\usepackage[color]{showkeys} 

\colorlet{refkey}{orange!20}
\colorlet{labelkey}{blue!30}

\numberwithin{equation}{section}
\newtheorem{theorem}{Theorem}[section]
\newtheorem{proposition}[theorem]{Proposition}
\newtheorem{lemma}[theorem]{Lemma}

\crefname{claim}{Claim}{Claims}

\newtheorem*{question*}{Question}

\theoremstyle{definition}

\newtheorem*{definition*}{Definition}

\theoremstyle{remark}
\newtheorem*{remark}{Remark}

\usepackage[linesnumbered,ruled,vlined]{algorithm2e}
\usepackage[noend]{algpseudocode}

\crefname{equation}{}{} 
\AtBeginEnvironment{appendices}{\crefalias{section}{appendix}} 
\crefname{algocfline}{Algorithm}{Algorithms}

\usepackage[color]{showkeys} 

\colorlet{refkey}{orange!20}
\colorlet{labelkey}{blue!30}

\newcommand{\mb}{\mathbb}

\newcommand{\mc}{\mathcal}

\newcommand{\ol}{\overline}

\newcommand{\tsc}{\textsc}

\newcommand{\init}{\textsc{in}}
\newcommand{\out}{\textsc{f}}
\newcommand{\assign}{\leftarrow}
\DeclareMathOperator{\Supp}{Supp}
\allowdisplaybreaks

\author{
Vishesh Jain \\
Stanford University \\
\texttt{vishesh.vj@gmail.com}
\and
Ashwin Sah \\
Massachusetts Institute of Technology \\
\texttt{asah@mit.edu}
\and
Mehtaab Sawhney \\
Massachusetts Institute of Technology \\
\texttt{msawhney@mit.edu}
}

\date{}
\begin{document}
\title{Perfectly Sampling $k\ge (8/3 +o(1))\Delta$-Colorings in Graphs}
\begin{titlepage}
\clearpage\maketitle
\thispagestyle{empty}
\begin{abstract}
We present a randomized algorithm which takes as input an undirected graph $G$ on $n$ vertices with maximum degree $\Delta$, and a number of colors $k \geq (8/3 + o_{\Delta}(1))\Delta$, and returns -- in expected time $\tilde{O}(n\Delta^{2}\log{k})$ -- a proper $k$-coloring of $G$ distributed \emph{perfectly} uniformly on the set of all proper $k$-colorings of $G$. Notably, our sampler breaks the barrier at $k = 3\Delta$ encountered in recent work of Bhandari and Chakraborty [STOC 2020]. We also sketch how to modify our methods to relax the restriction on $k$ to $k \geq (8/3 - \epsilon_0)\Delta$ for an absolute constant $\epsilon_0 > 0$.  

As in the work of Bhandari and Chakraborty, and the pioneering work of Huber [STOC 1998], our sampler is based on Coupling from the Past [Propp\&Wilson, Random Struct.\ Algorithms, 1995] and the bounding chain method [Huber, STOC 1998; H\"aggstr\"om\& Nelander, Scand.\ J.\ Statist., 1999]. Our innovations include a novel bounding chain routine inspired by Jerrum's analysis of the Glauber dynamics [Random Struct.\ Algorithms, 1995], as well as a preconditioning routine for bounding chains which uses the algorithmic Lov\'asz Local Lemma [Moser\&Tardos, J.ACM, 2010]. 
\end{abstract}
\end{titlepage}
\newpage
\section{Introduction}\label{sec:introduction}
Let $G = (V(G), E(G))$ be an undirected graph with vertex set $V(G)$ and edge set $E(G)$. For an integer $k \ge 1$, a \emph{(proper) $k$-coloring} of $G$ is a map $\varphi: V(G) \to [k](:=\{1,\dots, k\})$ such that for all $\{u,v\}\in E(G)$, $\varphi(u)\neq \varphi(v)$. 
In this paper, we study the problem of efficiently \emph{perfectly} sampling a $k$-coloring of a graph with maximum degree $\Delta$, uniformly at random from among all such colorings.

\subsection{Sampling \texorpdfstring{$k$}{k}-colorings: approximately and perfectly}
The algorithmic problem of sampling a uniformly random $k$-coloring of a graph with maximum degree $\Delta$ has been intensely studied (see, e.g., the references in \cite{FV07,CDMPP19}). Perhaps the major open problem in this area is to devise -- for all $k \geq \Delta + 2$ --  a (randomized) algorithm, with running time polynomial in $n = |V(G)|$, $k$, and $\ln(1/\epsilon)$, which outputs a distribution within total variation distance $\epsilon$ of the uniform distribution on the space of $k$-colorings; the lower bound corresponds to the minimum number of colors needed to ensure that the space of $k$-colorings is connected in a certain sense, and is within $1$ color of the classical theorem of Brooks which asserts that $\Delta+1$ colors are sufficient to color any graph of maximum degree $\Delta$ (and necessary for cliques and cycles of odd length). 

Recall that the Glauber dynamics on the space of $k$-colorings is the Markov chain which, at a coloring $\chi$, chooses a vertex uniformly at random from $V(G)$, and updates its color to be a color chosen uniformly at random from among those not already occupied by its neighbors; it is readily seen that this Markov chain is ergodic for $k\geq \Delta + 2$, and has the uniform distribution on the space of $k$-colorings as its stationary distribution. 
In a seminal work, Jerrum \cite{Jer95} showed that the Glauber dynamics mixes rapidly for $k > 2\Delta$, thereby providing an efficient algorithm for approximately sampling $k$-colorings for all $k > 2\Delta$. The lower bound on $k$ was relaxed by Vigoda \cite{Vig00} to $11\Delta/6$ by using a different Markov chain based on `flip dynamics', although by using comparison techniques, his proof also implies rapid mixing of the Glauber dynamics for $k > 11\Delta/6$. Recently, Chen, Delcourt, Moitra, Perarnau, and Postle \cite{CDMPP19} sharpened Vigoda's analysis to further relax the lower bound to $k > (11/6 - \epsilon_0)\Delta$, where $\epsilon_0$ is a small absolute constant $(\sim10^{-4})$. Under additional assumptions on the degree and girth of $G$, even less restrictive lower bounds on $k$ are known (see, e.g., the references in \cite{FV07,CDMPP19,BC20}).\\ 

The problem of efficiently (i.e.\ polynomial in $n$ and $k$) perfectly sampling $k$-colorings, which is the focus of this paper, was first studied by Huber \cite{Hub98}, who used Coupling from the Past (CFTP) \cite{PW96} along with  the bounding chain method \cite{Hag98,Hub98} to devise an efficient algorithm for perfectly sampling $k$-colorings, provided that $k > \Delta(\Delta + 2)$. 
One of the motivations of Huber's work was that using perfect sampling algorithms in the general sampling-to-counting framework of Jerrum, Valiant, and Vazirani \cite{JVV86} can potentially be used to obtain faster algorithms for the problem of \emph{approximately counting} the number of $k$-colorings of a graph, than can be obtained from approximate sampling algorithms \cite[Theorem~7]{Hub98}. Another motivation for his work was that -- in contrast to the approximate sampling algorithms discussed above, which always need to be run for the theoretical worst-case time in order to output a distribution guaranteed to be close to the uniform distribution -- perfect sampling algorithms based on CFTP have the attractive property of coming with a well-defined termination criterion, which may be reached in practice well before the time suggested by worst-case analysis (in fact, for implementing such an algorithm, the practitioner need not have \emph{any} knowledge of the worst-case running time).

In recent years, by exploiting the connection, due to Jerrum, Valiant, and Vazirani \cite[Theorem~3.3]{JVV86}, between \emph{deterministic} approximate counting and perfect sampling, several improvements of Huber's result, which are efficient for graphs of \emph{constant} maximum degree, have been obtained. Using the correlation decay technique, Gamarnik and Katz \cite{GK12}, respectively Lu and Yin \cite{LY13}, obtained perfect samplers with $k > 2.78\Delta$ for triangle free graphs, respectively $k > 2.58 \Delta$ for general bounded degree graphs, both running in time $O(n^{O(\log k)})$; using  Barvinok's polynomial interpolation method, Liu, Sinclair, and Srivastava \cite{LSS19} provided a perfect sampler for $k \ge 2\Delta$ running in time $O(n^{\exp(\text{poly}(k))})$.

For general graphs, where $\Delta$ (or $k$) is allowed to grow with $n$, the only improvement of Huber's result is the recent work of Bhandari and Chakraborty \cite{BC20}, who provided an efficient CFTP based perfect sampler for $k > 3\Delta$. The natural question left open by their work is whether one can devise efficient perfect samplers for $k < 3\Delta$ -- indeed, the sampler in \cite{BC20} may be viewed as implementing a two-stage process, with natural barriers at $k=3\Delta$ encountered (for different reasons) at both the stages (see \cref{sec:bc} for a quick overview and \cite[Section~3]{BC20} for a more detailed explanation).

\subsection{Our result}
As our main result, we obtain a CFTP based perfect sampler for $k > (8/3 + o(1))\Delta$. Pleasantly, our sampler has the same expected running time as in \cite{BC20}. 

\begin{theorem}
\label{thm:main}
There is a (CFTP-based) randomized algorithm $\tsc{PerfectSampler}$ (\cref{alg:perfect-sampler}) and an absolute constant $C_{\ref{thm:main}} > 0$ such that the following holds. Given an undirected graph $G$ with maximum degree $\Delta$, and a number of colors $k$ with $k \geq 8\Delta/3 + C_{\ref{thm:main}}\sqrt{\Delta\log{\Delta}}$, $\tsc{PerfectSampler}$ returns a uniformly random $k$-coloring of $G$, and runs in expected time $O(T_1 + T_2 + T_3)$, where $T_1,T_2,T_3 = O(n(\log n)^2\Delta^2(\log\Delta)(\log k))$.
\end{theorem}

\begin{remark}
The proof shows that for $\Delta$ sufficiently large, taking $C_{1.1} = 2$ is sufficient. Moreover, in \cref{sec:conclusion}, we briefly indicate how our sampler may be modified to obtain a version of \cref{thm:main} for $k \geq (8/3 - \epsilon)\Delta$, where $\epsilon \approx 10^{-2}$ is an absolute constant. We decided not to pursue this improvement since (i) the details are a bit more technical and all the main ideas are already present in our current analysis, (ii) the improvement is relatively minor, and can anyway not reach $k > 5\Delta/2$, which we believe is a natural barrier for our methods (see \cref{sec:conclusion} for a discussion of this). 
\end{remark}

\subsection{Organization} The rest of this paper is organized as follows. In \cref{sec:coupling}, we provide an introduction to coupling from the past, and the bounding chain method. In particular, the standard \cref{lemma:CFTP} reduces the proof of \cref{thm:main} to the construction of a certain procedure which we call $\tsc{SamplerUnit}$. In \cref{sec:outline}, we provide an overview of this procedure -- \cref{sub:notation} contains some notation used throughout the paper, \cref{sec:seeding-set} contains a preliminary routine used by our algorithm, whose proof is presented in \cref{sec:proof-find-seeding-set}, \cref{sec:bc} provides a quick introduction to the sampler in \cite{BC20}, \cref{sec:alg} provides a description of \tsc{SamplerUnit}, modulo the details of some primitive routines, and finally, \cref{sec:key-ideas} provides a high-level discussion of the key ideas underpinning the construction and analysis of our sampler. \cref{sec:routines} presents and analyses our main primitives -- $\tsc{compress}, \tsc{seeding}$, and $\tsc{disjoint}$, \cref{sec:analysis} completes the analysis of $\tsc{SamplerUpdate}$, and \cref{sec:conclusion} concludes with some final remarks (including a brief sketch of how to relax the lower bound in \cref{thm:main} to $(8/3 - \epsilon_0)\Delta$) and directions for future research. 
\section{Coupling from the past and bounding chains}\label{sec:coupling}
\subsection{Coupling from the past}
\label{sub:cftp}
As in \cite{Hub98,BC20}, our perfect sampler is based on coupling from the past (CFTP), which is a general procedure due to Propp and Wilson \cite{PW96} for sampling exactly from the stationary distribution of a Markov chain. The basic idea behind CFTP is that for an ergodic Markov chain started at time $-\infty$, its location at time $0$ should be distributed according to the stationary distribution; hence, if we could determine the location at time $0$ by only looking at the randomness generating the chain in the recent past, then we would have an efficient way of obtaining a sample from the stationary distribution of the chain. 

Implementing this idea algorithmically for an ergodic Markov chain on a finite state space $\Omega$ typically amounts to the following: for $i = 1,2,\dots, T$, we generate independent random maps $f_{-i}: \Omega \to \Omega$ with the property that if $\omega \in \Omega$ is distributed according to the stationary distribution, then $f_{-i}(\omega)$ is also distributed according to the stationary distribution. If it so happens that the composite function
\[F_{-1,-T} := f_{-1}\circ \dots \circ f_{-T}\]
is constant on $\Omega$, then we are guaranteed that $F_{-1,-T}(\omega_0)$ (for any $\omega_0 \in \Omega$; note that the image does not depend on the choice of $\omega_0$) is a sample from the stationary distribution. If $F_{-1,-T}$ is not constant, then we can consider $F_{-1,-T}\circ F_{-T-1, -2T}$ (by independently generating $f_{-T-1},\dots, f_{-2T}$), and so on. More formally, \cref{thm:main} follows from the following standard lemma, once we have constructed a suitable randomized algorithm $\tsc{SamplerUnit}$ and predicate $\Phi$. 
\begin{lemma}
\label{lemma:CFTP}
Let $G$ be an undirected graph on $n$ vertices with maximum degree $\Delta$, let $k \geq 8\Delta/3 + C_{\ref{thm:main}}\sqrt{\Delta\log\Delta}$, and let $\Omega$ denote the set of $k$-colorings of $G$. Suppose there is a randomized algorithm $\tsc{SamplerUnit}$ for generating a distribution $\mc{D}$ on functions $F: \Omega \to \Omega$, and a predicate $\Phi: \Supp(\mc{D}) \to \{\tsc{true}, \tsc{false}\}$ with the following properties:
\begin{enumerate}[(P1)]
    \item If $\chi$ is uniformly distributed in $\Omega$, and $F$ is generated according to $\mc{D}$ independently of $\chi$, then $F(\chi)$ is also uniformly distributed in $\Omega$. 
    \item If $\Phi(F) = \tsc{true}$, then $F$ is constant on $\Omega$.
    \item $\mb{P}_{F\sim \mc{D}}[\Phi(F) = \tsc{true}] \geq 1/2$.
    \item $\tsc{SamplerUnit}$ runs in time $T_1$, $\Phi(F)$ can be computed in time $T_2$, and $F(\chi)$ can be computed in time $T_3$. 
\end{enumerate}
Then, the randomized algorithm $\tsc{PerfectSampler}$ terminates in expected time $O(T_1 + T_2 + T_3)$ and returns a uniformly distributed $k$-coloring of $G$. 
\end{lemma}
\begin{proof}
Let $F_{-1}, F_{-2}, \dots$ be the i.i.d.\ samples from $\mc{D}$ generated by $\tsc{PerfectSampler}$. Let $\chi$ be an independent and uniformly distributed $k$-coloring, let $\chi_{i} = F_{-1}\circ \dots \circ F_{-i}(\chi)$, and let $\chi^*$ be the output of the algorithm. By (P1), it follows that for all $i\geq 1$, $\chi_{i}$ is also a uniformly distributed $k$-coloring. Moreover, by (P2, P3), $\chi_{i} = \chi^*$ with probability at least $1-2^{-i}$ (since this happens whenever $\vee_{j=1}^{i}\Phi(F_{-i})=\tsc{true}$). In particular, by the coupling characterization of total variation distance, $\chi^*$ is within total variation distance $2^{-i}$ of the uniform distribution on the space of $k$-colorings. Finally, since $i \geq 1$ is arbitrary, it follows that $\chi^*$ is actually itself uniformly distributed. 
The claim about the running time follows easily by noting that the outer loop is executed at most $2$ times in expectation.
\end{proof}
\begin{algorithm}[ht]\label{alg:perfect-sampler}
\caption{$\tsc{PerfectSampler}$ -- Takes an input procedure $\tsc{SamplerUnit}$ and converts procedure into a perfect sampler.}
Compute seeding set $\mc{S}$ (as in \cref{prop:seeding})\\
\For{$i = 1,2,\ldots$}{
Generate $F_{-i}$ according to $\tsc{SamplerUnit}$\\
\If{$\Phi(F_{-1}\circ \dots \circ F_{-i})=\tsc{True}$}{Output unique coloring in the image of $F_{-1}\circ \dots \circ F_{-i}$ and \tsc{Terminate}}
}

\end{algorithm}

The main challenge in CFTP based algorithms is efficiently determining whether $\Phi(F) = \tsc{true}$. \emph{A priori}, this requires evaluating $F$ for every $\omega \in \Omega$, which is infeasible if $|\Omega|$ is very large. However, in certain contexts where the domain $\Omega$ is equipped with a natural partial order compatible with the Markov chain, considerations of monotonicity or anti-monotonicity can reduce this task to evaluating $F$ on only a small number of `extremal' elements (see, e.g., \cite{PW96,Hag98} for examples). Unfortunately, in our case, where $\Omega$ is the space of $k$-colorings, $|\Omega|$ is too large (potentially $k^{n}$) to permit direct evaluation of $F$, and moreover, there doesn't seem to be any natural notion of (anti)monotonicity compatible with various Markov chains on the space of colorings.   

\subsection{Bounding chains} To overcome this issue, Huber \cite{Hub98} and independently H\"aggstr\"om and Nelander \cite{Hag98} introduced the method of bounding chains. The way this method is implemented in the case of $k$-colorings is the following: while evaluating $\Phi(F)$, where $F$ is the composite function $F_{-1,-T}$ as in the previous subsection, instead of precisely keeping track of the intermediate images $f_{-j}\circ \dots \circ f_{-T}$, we instead maintain a set $L_{-j+1}(v)$ of colors for each vertex $v \in V(G)$ with the property that for all $j \in [T]$, the image of $\Omega$ under $f_{-j}\circ \dots \circ f_{-T}$ is contained in $L_{-j+1}(v_1)\times \dots \times L_{-j+1}(v_n)$. Then, if we can show that $|L_{0}(v)|=1$ for all $v\in V(G)$, we will be done. The idea here is that the (product) space of sets  of available colors at each vertex, while cruder, is more amenable to the design of CFTP algorithms (for instance, note that there is a natural partial order on this space induced by set-theoretic inclusion of the set of available colors at each vertex). 
The perfect samplers in \cite{Hub98,BC20} are both based on CFTP and the bounding chain method. Our improvement stems from a novel implementation of this method (see \cref{sec:key-ideas} for a discussion of the key ideas); in particular, among other things, we find a way of lifting Jerrum's analysis \cite{Jer95} of the rapid mixing of Glauber dynamics to bounding chains (\cref{sec:lifting-jerrum}). 

\section{Overview of \tsc{SamplerUnit}}
\label{sec:outline}
\subsection{Notation}
\label{sub:notation}
Throughout, $G$ will be an undirected graph on $n$ vertices with maximum degree $\Delta$. A \emph{bounding list} is a list $L = (L(v): v\in V(G))$, where each $L(v)$ is a subset of colors in $[k]$. We will often refer to $L(v)$ as the \emph{bounding set} of the vertex $v$. Given a vertex $v\in V(G)$, we let
\[S_L(v) = \bigcup_{w\in N(v)}L(w),\qquad Q_L(v) = \bigcup_{\substack{w\in N(v)\\|L(w)|=1}}L(w).\]
Here, as is standard, $N(v)$ denotes the neighborhood of a vertex $v$. A key quantity in our algorithm is the set \[N_L^\ast(v) = \{w\in N(v): |L(w)|=2\text{ and }L(w)\cap L(w') = \emptyset\text{ if }w'\in N(w),w'\neq w\},\]
and the set of \emph{disjoint-pair colors} associated to $v$, defined by
\[D_L(v) = \bigcup_{w\in N_L^\ast(v)}L(w).\]
Finally, let
\[E_L(v) = S_L(v) \setminus (Q_L(v)\cup D_L(v)).\]

We will reserve the symbols $\chi, \chi'$ for $k$-colorings, and say that $\chi$ is compatible with $L$, denoted by $\chi \sim L$, if $\chi(v) \in L(v)$ for all $v$. As in \cite{BC20}, we will associate update operations with tuples -- specifically, we will use $6$-tuples of the form
\[\alpha = (v,\tau, L, L', M, \gamma),\]
where $v \in V(G)$, $\tau \in [0,1]$, $L,L'$ are bounding lists, $M$ is a sequence of at most $\Delta + 1$ distinct colors from $[k]$, and $\gamma \in [3]$ specifies the `type' of the update. We will denote the update operation (i.e.\ the map from the space of proper colorings to itself) associated to the tuple $\alpha$ by $f_{\alpha}$; in particular, the sequence of random functions $f_{-1},\dots, f_{-T},\dots $ discussed in \cref{sub:cftp} will be specified by the sequence of random updates $\alpha_{-1},\dots, \alpha_{-T},\dots$. 

As in \cite{BC20}, $\tsc{SamplerUnit}$ will consist of a sequence of $T$ updates satisfying the following three key properties. Fix $t \in \{-T,\dots, -1\}$. First, the random vertex $v_{t}$ is independent of $\alpha_{-T},\dots,\alpha_{-t-1}$. Second, $f_{\alpha_{t}}$ implements the Glauber dynamics at $v_{t}$ i.e.\ for any coloring $\chi$, $f_{\alpha_{t}}(\chi)(w) = \chi(w)$ for all $w \neq v_{t}$ and $f_{\alpha_{t}}(\chi)(v_{t})$ is uniformly distributed in $[k]\setminus \chi(N(v))$. Third, if $\chi \sim L_{t}$, then $f_{\alpha_{t}}(\chi) \sim L_{t}'$.

\subsection{Finding a seeding set}
\label{sec:seeding-set}
The very first step of $\tsc{PerfectSampler}$ consists of efficiently finding a set $\mc{S}$ of \emph{seeded vertices}, as defined in the following proposition. In fact, if we need to generate multiple samples, we can perform this step only once at the start, and use the same $\mc{S}$ for all calls to  $\tsc{PerfectSampler}$.  
\begin{proposition}\label{prop:seeding}
Fix $\eta \in (0,1/3)$ and let $\Delta\ge C_\eta$. There is a set of vertices $S\subseteq V(G)$ such that any $v\in V(G)$ satisfies
\begin{align}
\label{eq:seeding-conclusion}
|N(v)\cap\mc{S}^c|\le (1-\eta)\Delta,\qquad |N(v)\cap\mc{S}|\le\Delta/3.
\end{align}
Furthermore, there is a randomized algorithm that finds such a set with probability at least $1/2$ and runs in time $O(n \Delta + n\log n)$.
\end{proposition}

\begin{remark}
In fact, we can let $C_\eta$ be an absolute constant for $\eta = 1/3 - 2\sqrt{(\log\Delta)/\Delta}$.
\end{remark}

\cref{prop:seeding} is purely a statement about probabilistic combinatorics, and has nothing to do with graph colorings. Its proof is based on a standard application of the algorithmic Lov\'asz Local Lemma due to Moser and Tardos \cite{MT10}, and is included in \cref{sec:proof-find-seeding-set} for completeness.

\subsection{Outline of the Bhandari-Chakraborty construction}
\label{sec:bc}
Before presenting our construction of $\tsc{SamplerUnit}$, it is instructive to briefly review the salient features of the corresponding construction in \cite{BC20}; we refer the reader to \cite[Section~1.2]{BC20} for a more detailed overview. Recall that the goal of the block of $T$ updates in $\tsc{SamplerUnit}$ is to ensure that, with probability at least $1/2$, the bounding list at the most recent time consists of sets of size $1$. In \cite{BC20}, this is accomplished in two phases -- the first phase (referred to as `collapsing') serves to ensure that all bounding sets are of size at most $2$, whereas the second phase (referred to as `coalescing') makes all bounding sets of size $1$ with probability at least $1/2$. 

The two phases are themselves based on two types of updates, called $\tsc{compress}$ and $\tsc{contract}$. The coalescing phase consists of a predetermined number of applications of $\tsc{contract}$ at uniformly randomly chosen vertices. Whenever $\tsc{contract}$ is applied at a vertex $v$, it results in the bounding set at $v$ contracting to size at most $2$, and with some probability, to size $1$. However, to apply $\tsc{contract}$ at $v$, one needs the promise that $|S_{L}(v)| < k - \Delta$; since bounding sets of size $2$ (which is the only guarantee we have at the end of the collapsing phase) can in general lead to $|S_{L}(v)| = 2\Delta$, this is one source of the restriction $k > 3\Delta$. Also, while $\tsc{contract}$ leads to bounding sets of size at most $2$, it may very well happen that applying $\tsc{contract}$ to a vertex which already has bounding set of size $1$ leads to a larger bounding set of size $2$. During the coalescing phase, in order for the repeated random applications of $\tsc{contract}$ to lead to a `drift' towards all bounding sets having size $1$, one also needs the condition that $k > \Delta + \Delta\cdot {\max_{v\in V(G)}}|L(v)|$, which again leads to the restriction $k > 3\Delta$. 

In contrast to the coalescing phase, where the vertices are chosen uniformly at random, the vertices chosen to update in the collapsing phase are predetermined (this is one of the chief innovations of \cite{BC20}). Indeed, for an arbitrary ordering $v_1,\dots,v_{n}$ of the vertices, the collapsing phase can be concisely represented as \[\tsc{spruceup}(v_1), \tsc{contract}(v_1),\dots, \tsc{spruceup}(v_n), \tsc{contract}(v_n),\] where the job of $\tsc{spruceup}(v_i)$ is to ensure that the condition $|N_L(v_i)| < k - \Delta$, needed to apply $\tsc{contract}(v_i)$, is satisfied. Finally, $\tsc{spruceup}(v_i)$ is performed as follows: first, we pick an \emph{arbitrary} set $A$ of size $\Delta$ which non-trivially intersects the bounding sets of all neighbors of $v_i$ preceding $v_i$ (in the fixed ordering of vertices). Next, to each neighbor of $v_i$ succeeding it in the ordering, we apply $\tsc{compress}$ with input $A$ -- this has the effect of changing the bounding sets at these vertices to be the union of $A$ and a color outside of $A$. Note that once $\tsc{spruceup}(v_i)$ is completed, we indeed have $|N_L(v_i)| \leq 2\Delta$, since the vertices preceding $v_i$ have already been contracted (and hence, can contribute at most one color outside of $A$) whereas the vertices succeeding $v_i$ can also contribute at most one color outside of $A$ (by definition of $\tsc{compress}$).

\subsection{Our construction of \texorpdfstring{$\tsc{SamplerUnit}$}{SAMPLERUNIT}} 
\label{sec:alg}
We are now ready to present our construction, which is based on three kinds of updates -- $\tsc{compress}$ (\cref{alg:compress}), $\tsc{seeding}$ (\cref{alg:seeding}), and $\tsc{disjoint}$ (\cref{alg:disjoint}). Throughout this subsection, let $\eta = 1/3 - 2\sqrt{(\log \Delta)/\Delta}$, let $k > (3-\eta)\Delta$, and let $\mc{S}$ be the set of vertices coming from \cref{prop:seeding} applied with $\eta $. For the sake of simplicity, let $s = |\mc{S}|$. Also, throughout the rest of this paper, we assume that $\Delta \geq C$, for some sufficiently large absolute constant $C$. This may be done without loss of generality since, in \cref{thm:main}, the constant $C_{\ref{thm:main}}$ may be taken sufficiently large so that for $\Delta \leq C$, the lower bound on $k$ is $k > 3\Delta$, at which point one may use either the sampler in \cite{BC20}, or indeed our sampler with a slightly more careful analysis.\\

We construct $\tsc{SamplerUnit}$ in the following four phases:

\begin{itemize}
    \item Phase $1$ (Seeding step): Arbitrarily order the vertices in $\mc{S}$ as $v_1,\ldots, v_s$. For $1\le i\le s$, perform $\tsc{compress}$ on all neighbors of $v_i$ that are not in $\{v_1,\ldots,v_{i-1}\}$ with associated set $A$ being an arbitrary set of size $\Delta$ \emph{completely containing} $L(w)$ for each $w\in N(v_i)\cap \{v_j: j < i\}$. Then, perform $\tsc{seeding}$ on $v_i$ and increment $i$ by $1$ (if $i < s$) or move to Phase 2 (if $i=s$). Note that at the end of this phase, all vertices in $\mc{S}$ have bounding set of size at most $3$.
    \item Phase $2$ (Converting seeded vertices to size $2$): For each $1\le i\le s$, apply $\tsc{compress}$ to all neighbors of $v_i$ not in $\mc{S}$, with associated set $A$ being an arbitrary set of size $\Delta$ \emph{completely containing} $N(w)$ for all $w\in N(v_i)\cap\mc{S}$. Then, apply $\tsc{disjoint}$ to $v_i$. Note that at the end of this phase, all vertices in $\mc{S}$ have bounding set of size at most $2$. 
    \item Phase $3$ (Converting remaining vertices to size $2$): Mark all vertices in $\mc{S}$. Arbitrarily order the vertices in $V(G)\setminus\mc{S}$ as $v_{s+1},\dots, v_{n}$. For $s+1 \le i \le n$ perform the following sequence of operations. Apply $\tsc{compress}$ to all unmarked neighbors of $v_i$ with associated set $A$ of size $\Delta$ determined as follows: let $L_m$ be the current bounding list, restricted to marked neighbors of $v$. We greedily take elements from $Q_{L_m}(v)\cup E_{L_m}(v)$ first, then (if the set constructed at this point has size less than $\Delta$ colors) elements from $D_{L_m}(v)$ (chosen in pairs $L_m(w)$ for $w\in N_{L_m}^\ast(v)$), and then (if we still do not have $\Delta$ colors) arbitrarily from the remaining colors. Then, apply $\tsc{disjoint}$ to $v_i$, mark $v_i$, and increment $i$ by $1$ (if $i < n$) or move to Phase 4 (if $i = n$). Note that at the end of this phase, all vertices have bounding set of size at most $2$. 
    \item Phase $4$ (Drifting to size $1$): For $T_{D} = 2(k-\Delta)n\log{n}/(k-5\Delta/2)$ steps, apply $\tsc{disjoint}$ on a uniformly random vertex in the graph.
\end{itemize}

In \cref{sec:analysis}, we show how $\tsc{SamplerUnit}$ can be used to generate a distribution $\mc{D}$ and a predicate $\Phi$ satisfying (P1)-(P4) in \cref{lemma:CFTP} with $T_1, T_2, T_3$ as in \cref{thm:main}.

\subsection{Key ideas}
\label{sec:key-ideas}
In this subsection, we provide an informal and high-level discussion of some of the key ideas underlying our construction and analysis of $\tsc{SamplerUnit}$.  
\subsubsection{Lifting Jerrum's analysis of Glauber dynamics using \texorpdfstring{$D_L(v)$}{DL(v)}}
\label{sec:lifting-jerrum}
The main idea which enables us to bypass the obstacle at $3\Delta$ encountered in the coalescing phase is a bounding list version of Jerrum's analysis in \cite{Jer95}. Recall that in the standard (path) coupling argument proving rapid mixing of Glauber dynamics for $k > 3\Delta$, one couples two chains by generating a uniformly random pair $(v,c)\in V(G)\times [k]$, and updating the color at $v$ to $c$ whenever possible. Suppose we have two colorings $\chi$ and $\chi'$ differing only at a single vertex $v_0$, with $\chi(v_0) = c_0$ and $\chi'(v_0) = c_0'$. Then, under this coupling, the distance between $\chi$ and $\chi'$ decreases by $1$ iff for the random pair $(v,c)$, $v = v_0$ and $c$ is one of the at most $k-\Delta$ colors not appearing in $|\chi(N(v_0))|$ (note that $\chi(N(v_0)) = \chi'(N(v_0))$. Also, the distance between $\chi$ and $\chi'$ increases by $1$ iff for the random pair $(v,c)$, $v \in N(v_0)$ and $c \in \{c_0, c_0'\}$. Since there are at least $k-\Delta$ pairs which decrease the distance by $1$, and at most $2\Delta$ pairs which increase the distance by $1$, we see (at least intuitively) that the distance drifts towards $0$ if $k - \Delta > 2\Delta$ i.e.\ $k > 3\Delta$. Jerrum improved the lower bound to $2\Delta$ by slightly modifying this coupling so that whenever $c_0$ (respectively $c_0'$) is selected by for $\chi$, $c_0'$ (respectively $c_0$) is selected for $\chi'$; it is immediate that this halves the number of `bad' pairs $(v,c)$, and leads to the weaker restriction $k - \Delta > \Delta$ i.e.\ $k > 2\Delta$. 

In our algorithm, we perform a similar coupling of the colors in $D_L(v)$ (which are naturally paired up by definition) -- this appears as one of the cases in the update $\tsc{disjoint}$. However, in order to obtain any improvement via such a coupling, we need $|D_L(v)|$ to constitute a non-trivial fraction of $|N_L(v)|$, which need not be the case (note that no such difficulty arises in Jerrum's work). We overcome this issue using a win-win analysis based on a robust version (see \cref{eq:weights}) of the extremal combinatorial fact that if $|S_L(v)| > 3\Delta/2,$ then $|D_L(v)| \neq 0$.\\ 

While this idea takes care of the barrier at $3\Delta$ \emph{provided} the bounding list consists of sets of size at most $2$, getting to this stage presents a different obstacle owing to the fact that the update $\tsc{contract}$ in \cite{BC20} requires $k > 3\Delta$ in the worst case to satisfy its promise. We circumvent this issue by using a combination of several ideas. 

\subsubsection{Preconditioning via \tsc{Seeding}} In contrast to \cite{BC20, Hub98}, we make much greater use of the structure of the underlying graph $G$ by first identifying a set $\mc{S}$ of size $\approx n/3$ such that each vertex has no more than $\approx 2\Delta/3$ neighbors outside $\mc{S}$ and no more than $\Delta/3$ neighbors inside $S$. Phase 1 ensures that all vertices in $\mc{S}$ have bounding sets of size at most $3$ -- in order to accomplish this, we introduce a new update called $\tsc{seeding}$ which requires a weaker promise than $\tsc{contract}$, but comes at the cost of the bounding set being of size at most $3$ (as opposed to $2$). Specifically, $\tsc{seeding}$ requires the guarantee that $k - \Delta \geq |S_L(v)|^{2}/(\Delta + |S_L(v)|)$; when $|S_L(v)| \leq 2\Delta$ (as can be guaranteed by applying $\tsc{compress}$ updates as in \cite{BC20}), the right hand side is at most $4\Delta/3$, so that the restriction on $k$ is only $k > 7\Delta/3$. 

\subsubsection{Substantially exploiting the flexibility in the choice of \texorpdfstring{$A$}{A}}
In \cite{BC20}, the set $A$ used for $\tsc{compress}$ updates to `spruce-up' the neighborhood of $v$ is always chosen to be simply a set of size $\Delta$ intersecting the bounding set of each neighbor of $v$ preceding it in the order, and has no additional properties. In contrast, our construction of $A$ is much more careful, and in fact, varies across phases to account for the different nature of the challenges encountered. In particular, when `sprucing-up' a vertex $v \in \mc{S}$ in Phase 2, we take $A$ to be a set of size $\Delta$ containing \emph{all} the colors appearing in any bounding set of $N(v) \cap \mc{S}$ -- by Phase 1 and the definition of $\mc{S}$, this is always possible. Then, note that applying $\tsc{compress}$ to the at most $\approx 2\Delta/3$ neighbors of $v$ not in $\mc{S}$ can contribute at most one additional color each, so that after this sprucing-up procedure, $|S_L(v)| \leq \Delta + \approx 2\Delta/3 = \approx 5\Delta/3$. At this point, we could use the $\tsc{contract}$ update from \cite{BC20} to convert the bounding set of $v$ to size at most $2$, but for a streamlined treatment, we use our more refined $\tsc{disjoint}$ update.  
\subsubsection{Refining \texorpdfstring{$\tsc{contract}$}{CONTRACT} by tracking \texorpdfstring{$D_L(v)$}{DL(v)}}Phase 3 of our algorithm, whose analysis is the most involved, combines the previous idea of exploiting the flexibility in the choice of $A$ with a variation of $\tsc{contract}$, called $\tsc{disjoint}$, which implements the idea in \cref{sec:lifting-jerrum}. Notably, as compared to $\tsc{contract}$, which requires the promise $k - \Delta > |S_L(v)|$, $\tsc{disjoint}$ requires the more refined promise
\[|S_L(v)| - |Q_L(v)| < (k-\Delta)\left(\frac{k-|Q_L(v)|}{k-|Q_L(v)| - |D_L(v)|/2}\right);\]
note that when $|Q_L(v)| = 0 = |D_L(v)|$, this reduces to the promise required by $\tsc{contract}$. Once the desired properties of the $\tsc{disjoint}$ update have been established (\cref{lem:disjoint}), the analysis of Phase 3 boils down to checking that the promise required by $\tsc{disjoint}$ is always satisfied; we show that by using a more intricate procedure for selecting the set $A$, this refined promise can be satisfied with around $8\Delta/3$ colors.





\section{\texorpdfstring{$\tsc{compress}$}{COMPRESS}, \texorpdfstring{$\tsc{seeding}$}{SEEDING}, and \texorpdfstring{$\tsc{disjoint}$}{DISJOINT}}\label{sec:routines}
In this section, we present and analyse our three main updates -- \tsc{compress}, \tsc{seeding}, and \tsc{disjoint}. 
In each case, we explain how the update is generated, how it interacts with the bounding list, and how one can apply the resulting random function to colorings (i.e.\ `decode the update') in order to simulate the Glauber dynamics at the appropriate vertex. 

\subsection{\texorpdfstring{$\tsc{compress}$}{COMPRESS}}\label{sub:compress}
The first update, $\tsc{compress}$, is exactly the same as in \cite{BC20}, which in turn builds on ideas in \cite{Hub98}; we sketch the analysis, as it serves as a warm-up for the analysis of our other updates. We define \tsc{Compress} in \cref{alg:compress}, and summarize its important properties in the following lemma. 
\begin{algorithm}[ht]\label{alg:compress}
\caption{$\tsc{compress}$ -- Takes an input update $\alpha_\init = (v_\init,\tau_\init,L_\init,L_\init',M_\init,\gamma_\init)$, a vertex $v$, and a set $A$ of size $\Delta$, and outputs a compatible ``compressed update''.}
\SetKwInOut{KwIn}{Input}
\SetKwInOut{KwOut}{Output}
\SetKwProg{myproc}{Function}{}{}
\myproc{\tsc{compress.gen}\rm{:}}{
 
\KwIn{$\alpha_\init=(v_\init,\tau_\init,L_\init,L'_\init, M_\init,\gamma_\init)$, $v\in V(G)$ and $A\subseteq [k]$ with $|A|=\Delta$}
\KwOut{$\alpha_\out=(v_\out,\tau_\out,L_\out,L'_\out,M_\out,\gamma_\out)$}
$\gamma_\out \assign 1$; $v_\out \assign v$; $L_\out \assign L_\init'$; $L_\out' \assign L_\out$; $\tau_\out \assign \tsc{unif}[0,1]$;\\
$c_1 \assign \tsc{unif}([k]\setminus A)$; $L_\out'(v) \assign A\cup\{c_1\}$; \\
$M_\out \assign \tsc{UnifPermutation}(A)$; $M_\out \assign (M_\out, c_1)$;\\
$\alpha_\out \assign (v_\out,\tau_\out,L_\out,L_\out',M_\out,\gamma_\out)$;
}
\myproc{\tsc{compress.decode}\rm{:}}{
 
\KwIn{$\alpha=(v,\tau,L,L', M,\gamma)$ with $\gamma = 1$ and a coloring $\chi\sim L$}
\KwOut{$\chi'\sim L'$}
$\chi' \assign \chi$;\\
$p_\chi(v) \assign \frac{k-\Delta}{k-|\chi(N(v))|}$;\\
$c_1 \assign M[\Delta+1]$\Comment{Since $\gamma = 1$, $M$ has length $\Delta+1$}\\
\If{$c_1\notin\chi(N(v))$\emph{ and }$\tau \le p_\chi(v)$}{
$\chi'(v) \assign c_1$;
} \Else {
$M' \assign M[1,\Delta]\setminus\chi(N(v))$; \\
$\chi'(v) \assign M'[1]$;\Comment{Exists if $c_1\in\chi(N(v))$}
}

}
\end{algorithm}

\begin{lemma}[{\cite[Lemma~2.1]{BC20}}]\label{lem:compress}
Let $k\ge\Delta + 1$. Let $\alpha_\init = (v_\init,\tau_\init,L_\init,L_\init',M_\init,\gamma_\init)$, and choose $v\in V(G)$ and $A\subseteq [k]$ with $|A| = \Delta$. Let $\alpha_\out = (v_\out,\tau_\out,L_\out,L_\out',M_\out',\gamma_\out')$ be the output of $\tsc{compress.gen}[\alpha_\init,v,A]$. Let $\chi$ be a coloring, and  let $\chi' = \tsc{compress.decode}[\alpha_\out,\chi]$. Then:
\begin{enumerate}[(C1)]
    \item $L_\out = L_\init'$, $L_\out'(u) = L_\out(u)$ for $u\neq v$, and $L_\out'(v) = A\cup\{c_1\}$ for some $c_1\in[k]\setminus A$.
    \item If $\chi\sim L_\out$, then $\chi'\sim L_\out'$.
    \item For $\chi\sim L_\out$, the random variable $\chi'$ is uniformly distributed over the set of colorings satisfying $\chi'(w) = \chi(w)$ for $w\neq v$ (i.e., this follows Glauber dynamics).
    \item Other than copying $L_\init'$, the expected runtime of $\tsc{compress.gen}$ is $O(\Delta\log k + \log n)$. The runtime of $\tsc{compress.decode}$ is $O(\Delta(\log\Delta\log k + \log n))$.
\end{enumerate}
\end{lemma}
\begin{proof}[Proof Sketch]
The first and second items follow trivially by construction, and the final item can also be justified easily (see \cite[Lemma~2.1]{BC20} for details); the technical heart of the above lemma is the third item, whose proof we now sketch.

Note that the randomness in $\chi'$ comes entirely from the $\tsc{compress.gen}$ routine. Consider some $\chi\sim L_\out = L_\init'$. $\tsc{compress.gen}$ chooses $c_1\in[k]\setminus A$ uniformly. $\tsc{compress.decode}$ changes only the color of $\chi$ at $v$, in the following way: if $c_1\notin\chi(N(v))$, then we let $\chi'(v) = c_1$ with probability $p_\chi(v)$. In all other cases, we let $\chi'(v)$ be a uniform color in $A\setminus\chi(N(v))$. Note that this set is empty only when $\chi(N(v)) = A$, which implies that $p_\chi(v) = (k-\Delta)/(k-|\chi(N(v))|) = 1$ and $c_1\notin\chi(N(v))$, i.e., that the first case is always invoked. Hence the decoding algorithm is well-defined.

Finally, we check that the color $\chi'(v)$ is chosen with the correct probability. For this, note that we choose any fixed element $c\in [k]\setminus(A\cup\chi(N(v)))$ if $c_1 = c$ and $\tau \le p_\chi(v)$, which happens with probability
\[\frac{1}{k-\Delta}\cdot\frac{k-\Delta}{k-|\chi(N(v))|} = \frac{1}{k-|\chi(N(v))|},\]
which is the correct probability according to the Glauber dynamics. By symmetry, the remaining probability is easily seen to be split equally among $A\setminus\chi(N(v))$, hence the  probability distribution of $\chi'(v)$ is indeed uniform on $[k]\setminus\chi(N(v))$.
\end{proof}

\subsection{\texorpdfstring{$\tsc{seeding}$}{SEEDING}}\label{sub:seeding}
The second update, $\tsc{seeding}$, is a variant of $\tsc{contract}$ in \cite{BC20}, and has the crucial property of operating under a weaker guarantee than $|S_L(v)|\le k-\Delta$. The tradeoff in exchange for this weaker guarantee is that the bounding set is no longer guaranteed to be of size $2$ but will instead be of size at most $3$.

\begin{algorithm}[ht]\label{alg:seeding}
\caption{$\tsc{seeding}$ -- Takes an input update $\alpha_\init = (v_\init,\tau_\init,L_\init,L_\init',M_\init,\gamma_\init)$ and a vertex $v$ with $\frac{|S_L(v)|^2}{\Delta + |S_L(v)|}\le k - \Delta$, and outputs a compatible ``seeding update''.}
\SetKwInOut{KwIn}{Input}
\SetKwInOut{KwOut}{Output}
\SetKwProg{myproc}{Function}{}{}
\myproc{\tsc{seeding.gen}\rm{:}}{
 
\KwIn{$\alpha_\init=(v_\init,\tau_\init,L_\init,L'_\init, M_\init,\gamma_\init)$, $v\in V(G)$ with $\frac{|S_L(v)|^2}{\Delta + |S_L(v)|}\le k - \Delta$}
\KwOut{$\alpha_\out=(v_\out,\tau_\out,L_\out,L'_\out,M_\out,\gamma_\out)$}
$\gamma_\out \assign 2$; $v_\out \assign v$; $L_\out \assign L_\init'$; $L_\out' \assign L_\out$; $\tau_\out \assign \tsc{unif}[0,1]$; \\
$c_1 \assign \tsc{unif}([k]\setminus S_L(v))$; $c_2,c_3 \assign \tsc{unif}(S_L(v))$; $L_\out'(v) \assign \{c_1,c_2,c_3\}$;\\
\Comment{$c_2,c_3$ chosen with repetition}\\
$M_\out \assign (c_1,c_2,c_3)$;\\
$\alpha_\out \assign (v_\out,\tau_\out,L_\out,L_\out',M_\out,\gamma_\out)$;
}
\myproc{\tsc{seeding.decode}\rm{:}}{
 
\KwIn{$\alpha=(v,\tau,L,L', M,\gamma)$ with $\gamma = 2$, $\frac{|S_L(v)|^2}{\Delta + |S_L(v)|}\le k - \Delta$, $M[1]\notin S_L(v)$, $M[2],M[3]\in S_L(v)$, and a coloring $\chi\sim L$}
\KwOut{$\chi'\sim L'$}
$\chi' \assign \chi$;\\
$p_\chi(v) \assign \frac{|S_L(v)|^2}{(k-|\chi(N(v))|)(|\chi(N(v))|+|S_L(v)|)}$;\\
$(c_1,c_2,c_3) \assign M[1,3]$\\
\If{$\{c_2,c_3\}\subseteq\chi(N(v))$\emph{ or }$\tau > p_\chi(v)$}{
$\chi'(v) \assign c_1$;
} \ElseIf{$c_2\notin\chi(N(v))$}{
$\chi'(v) \assign c_2$;
} \Else {
$\chi'(v) \assign c_3$;
}

}
\end{algorithm}

\begin{lemma}\label{lem:seeding}
Let $\alpha_\init = (v_\init,\tau_\init,L_\init,L_\init',M_\init,\gamma_\init)$, and choose $v\in V(G)$ such that $\frac{|S_L(v)|^2}{\Delta + |S_L(v)|}\le k - \Delta$. Let $\alpha_\out = (v_\out,\tau_\out,L_\out,L_\out',M_\out',\gamma_\out')$ be the output of $\tsc{seeding.gen}[\alpha_\init,v]$. Let $\chi$ be a coloring, and let $\chi' = \tsc{seeding.decode}[\alpha_\out,\chi]$. Then:
\begin{enumerate}[(S1)]
    \item $L_\out = L_\init'$, $L_\out'(u) = L_\out(u)$ for $u\neq v$, and $|L_\out'(v)| \le 3$.
    \item If $\chi\sim L_\out$, then $\chi'\sim L_\out'$.
    \item For $\chi\sim L_\out$, the random variable $\chi'$ is uniformly distributed over the set of colorings satisfying $\chi'(w) = \chi(w)$ for $w\neq v$ (i.e., this follows Glauber dynamics).
    \item Other than copying $L_\init'$, the expected runtime of $\tsc{seeding.gen}$ is $O(\Delta(\log k + \log n))$. The runtime of $\tsc{seeding.decode}$ is $O(\Delta(\log k + \log n))$.
\end{enumerate}
\end{lemma}
\begin{proof}
The first two items are trivial, and the fourth item follows in the same way as in \cite[Lemma~2.2(d)]{BC20}. We now verify the key third item.  First, note that $p_\chi(v)\in[0,1]$ by the condition given. Indeed, since $|\chi(N(v))|\in(0,\Delta]$ we easily see that $p_\chi(v)\ge 0$, and moreover, that $p_\chi(v)$ is maximized by its value when $|\chi(N(v))| = 0$ or $|\chi(N(v))| = \Delta$ (since the denominator of $p_{\chi}(v)$ is a concave function of $|\chi(N(v))|$ on $[0,\Delta]$). In the former case, we have $p_\chi(v) = |S_L(v)|/k\le 1$. In the latter case, we have
\[p_\chi(v) = \frac{|S_L(v)|^2}{(k-\Delta)(\Delta+|S_L(v)|)}\le 1\]
by the given condition.

Finally, consider any color $c\in S_L(v)\setminus\chi(N(v))$. It is chosen if and only if $\tau\le p_\chi(v)$, and also either $c_2 = c$ or $c_2\in\chi(N(v))$ and $c_3 = c$ (and the latter case clearly satisfies $c_2\neq c$). This occurs with probability
\[\bigg(\frac{1}{|S_L(v)|} + \frac{|\chi(N(v))|}{|S_L(v)|}\cdot\frac{1}{|S_L(v)|}\bigg)\cdot\frac{|S_L(v)|^2}{(k-|\chi(N(v))|)(|\chi(N(v))| + |S_L(v)|)} = \frac{1}{k - |\chi(N(v))|},\]
as required. Furthermore, since colors in $[k]\setminus S_L(v)$ are symmetrically chosen, the result again immediately follows.
\end{proof}

\subsection{\texorpdfstring{$\tsc{disjoint}$}{DISJOINT}}\label{sub:disjoint}
We now define our most complicated update, $\tsc{disjoint}$, which can be seen as combining the $\tsc{contract}$ update in \cite{BC20} with a bounding list version of Jerrum's analysis of the Glauber dynamics in \cite{Jer95}, by pairing up colors in $D_L(v)$. This additional pairing, compared to the analysis in \cite{BC20, Hub98}, is critical in obtaining a better drift estimate in the final coalescence phase and ensuring that the final stages of Phase $3$ succeed for $k<3\Delta$.
\bigskip

\begin{algorithm}[H]\label{alg:disjoint}
\caption{$\tsc{disjoint}$ -- Takes an input update $\alpha_\init = (v_\init,\tau_\init,L_\init,L_\init',M_\init,\gamma_\init)$ and a vertex $v$ with $S-Q< (k - \Delta)(\frac{k-Q}{k-Q-D/2})$, and outputs a compatible ``disjoint update''.}
\SetKwInOut{KwIn}{Input}
\SetKwInOut{KwOut}{Output}
\SetKwProg{myproc}{Function}{}{}
\myproc{\tsc{disjoint.gen}\rm{:}}{
 
\KwIn{$\alpha_\init=(v_\init,\tau_\init,L_\init,L'_\init, M_\init,\gamma_\init)$, $v\in V(G)$ with $S-Q < (k - \Delta)(\frac{k-Q}{k-Q-D/2})$}
\KwOut{$\alpha_\out=(v_\out,\tau_\out,L_\out,L'_\out,M_\out,\gamma_\out)$}
$\gamma_\out \assign 3$; $v_\out \assign v$; $L_\out \assign L_\init'$; $L_\out' \assign L_\out$; $\tau_\out \assign \tsc{unif}[0,1]$;\\
\If{$\tsc{unif}[0,1] > \frac{k-Q-D}{k-Q-D/2}$} {
$w \assign \tsc{unif}(N_L^\ast(v))$; $L_\out'(v) \assign N(w)$; $M_\out \assign N(w)$; \Comment{Arbitrarily order $N(w)$}
} \Else {
$c_1 \assign \tsc{unif}([k]\setminus S_L(v))$; $c_2 \assign \tsc{unif}(D_L(v))$;\\
$q(v) \assign 1 - \frac{(k-Q-D/2)E}{(k-Q-D)(k-\Delta)}$; $p_\Delta(v) \assign \frac{(\Delta-Q-D/2)D}{(k-\Delta)(k-Q-D)q(v)}$\\
\If{$\tsc{unif}[0,1]\le q(v)$} {
\If{$\tsc{unif}[0,1] > p_\Delta(v)$} {
$L_\out'(v) \assign \{c_1\}$; $M_\out \assign (c_1)$;
} \Else {
$L_\out'(v) \assign \{c_1,c_2\}$; $M_\out \assign (c_1,c_2)$;
}
} \Else {
$c_2 \assign \tsc{unif}(E_L(v))$; $L_\out'(v) \assign \{c_1,c_2\}$; $M_\out \assign (c_1,c_2)$;
}
}
$\alpha_\out \assign (v_\out,\tau_\out,L_\out,L_\out',M_\out,\gamma_\out)$;
}
\myproc{\tsc{disjoint.decode}\rm{:}}{
 
\KwIn{$\alpha=(v,\tau,L,L', M,\gamma)$ with $\gamma = 3$, $S-Q < (k - \Delta)(\frac{k-Q}{k-Q-D/2})$, and a coloring $\chi\sim L$}
\KwOut{$\chi'\sim L'$}
$\chi' \assign \chi$;\\
$q(v)\assign 1 - \frac{(k-Q-D/2)E}{(k-Q-D)(k-\Delta)}$; $p_\chi(v) \assign \frac{(|\chi(N(v))| - Q - D/2)D}{(k-|\chi(N(v))|)(k-Q-D)q(v)}$; $p_\chi'(v) \assign \frac{k-\Delta}{k-|\chi(N(v))|}$;\\
\If{$|M| = 1$\emph{ or }$M[1,2]\subseteq D_L(v)$} {
$\chi'(v) \assign M\setminus\chi(N(v))$;\Comment{This is size $1$ due to the disjointness condition}
} \Else {
\If{$M[2]\in D_L(v)$} {
$r_\chi(v) \assign p_\chi(v)/p_\Delta(v)$;
} \Else {
$r_\chi(v) \assign p_\chi'(v)$;
}
\If{$M[2]\in\chi(N(v))$\emph{ or }$\tau > r_\chi(v)$}{
$\chi'(v) \assign M[1]$;
} \Else {
$\chi'(v) \assign M[2]$;
}
}

}
\end{algorithm}
\bigskip
\bigskip
\bigskip
To begin, recall from \cref{sub:notation} that $S_L(v)$ is the set of colors appearing in the bounding lists $L(w)$ for neighbors $w$ of $v$, $Q_L(v)$ is the set of colors that appear in some bounding list $L(w)$ for $w\in N(v)$ with $|L(w)| = 1$, i.e., the bounding list forces this color to appear in $\chi(N(v))$ if $\chi\sim L$,  
\[N_L^\ast(v) = \{w\in N(v): |L(w)|=2\text{ and }L(w)\cap L(w') = \emptyset\text{ if }w'\in N(w),w'\neq w\},\]
and the disjoint-pair colors associated to $v$ are
\[D_L(v) = \bigcup_{w\in N_L^\ast(v)}L(w).\]
The key property of the disjoint-pair colors is that they appear in exactly one bounding set of a neighbor of $v$, and moreover, are naturally paired up with another disjoint-pair color via the bounding set of the \emph{same} element of $N_L^\ast(v)$. 
Note that $Q_L(v)\cap D_L(v) = \emptyset$ and $D_L(v)$ is a disjoint union of pairs $L(w)$ for $w\in N_L^\ast(v)$; in particular, $\chi(N(v))$ always has at least $|Q_L(v)|+|D_L(v)|/2$ different colors. Recall also that $E_L(v) = S_L(v)\setminus (Q_L(v)\cup D_L(v))$.

Finally, for the sake of notational lightness, throughout this subsection and in the definition of \cref{alg:disjoint}, we let $S = |S_L(v)|$, $Q = |Q_L(v)|$, $D = |D_L(v)|$, and $E = |E_L(v)| = S - Q - D$.

\begin{lemma}\label{lem:disjoint}
Let $\alpha_\init = (v_\init,\tau_\init,L_\init,L_\init',M_\init,\gamma_\init)$, and choose $v\in V(G)$ such that $S-Q < (k - \Delta)(\frac{k-Q}{k-Q-D/2})$. Let $\alpha_\out = (v_\out,\tau_\out,L_\out,L_\out',M_\out',\gamma_\out')$ be the output of $\tsc{disjoint.gen}[\alpha_\init,v]$. Let $\chi$ be a coloring and let $\chi' = \tsc{disjoint.decode}[\alpha_\out,\chi]$. Then:
\begin{enumerate}[(D1)]
    \item $L_\out = L_\init'$, $L_\out'(u) = L_\out(u)$ for $u\neq v$, and $|L_\out'(v)|\le 2$. Moreover, $|L_\out'(v)| = 1$ with probability $1 - \frac{S-Q}{k-\Delta}+\frac{D/2}{k-Q-D/2}$.
    \item If $\chi\sim L_\out$, then $\chi'\sim L_\out'$.
    \item For $\chi\sim L_\out$, the random variable $\chi'$ is uniformly distributed over the set of colorings satisfying $\chi'(w) = \chi(w)$ for $w\neq v$ (i.e., this follows Glauber dynamics).
    \item Other than copying $L_\init'$, the expected runtime of $\tsc{disjoint.gen}$ is $O(\Delta(\log k + \log n))$. The runtime of $\tsc{disjoint.decode}$ is $O(\Delta(\log k + \log n))$.
\end{enumerate}
\end{lemma}
\begin{proof}
To begin, we check that the quantities $q(v), p_{\chi}(v), p_{\Delta}(v), p'_{\chi}(v)$ appearing in \cref{alg:disjoint} lie in $[0,1]$ and $p_{\chi}(v)\le p_{\Delta}(v)$. Since $|\chi(N(v))|\le \Delta < k$, it follows that $p'_{\chi}(v)\in [0,1]$. Next, since
\begin{equation}\label{eq:disjoint-pchi}
\frac{k-Q-D}{k-Q-D/2}\cdot(1-q(v))\cdot\frac{1}{E}\cdot p_\chi'(v) = \frac{1}{k-|\chi(N(v))|},
\end{equation}
it follows that $1 - q(v) > 0$. 
Also, 
\begin{align*}
    \frac{1}{k-Q-D/2} + \frac{k-Q-D}{k-Q-D/2}\cdot q\cdot\frac{1}{D}
    &= \frac{k-Q}{(K-Q-D/2)D} - \frac{E}{(k-\Delta)D}\\
    &= \frac{k-Q}{(k-Q-D/2)D} - \frac{S-Q-D}{(k-\Delta)D}\\
    &> \frac{S-Q}{(k-\Delta)D} - \frac{S-Q-D}{(k-\Delta)D}\\
    &= \frac{1}{k-\Delta} \\
    &\ge \frac{1}{k-|\chi(N(v))|}\\
    &\ge \frac{1}{k-Q-D/2};
\end{align*}
where the strictly inequality uses our assumption that $S-Q < (k - \Delta)(\frac{k-Q}{k-Q-D/2})$; this shows that $q(v) > 0$. Since   
\begin{equation}\label{eq:disjoint-pchi'}
\frac{1}{k-Q-D/2} + \frac{k-Q-D}{k-Q-D/2}\cdot q(v)\cdot\frac{1}{D}\cdot p_\chi(v) = \frac{1}{k-|\chi(N(v))|},
\end{equation}
combining with the previous inequality shows that $p_{\chi}(v) \in [0,1]$. A similar argument also shows that $p_{\Delta} \in [0,1]$. Finally, since $|\chi(N(v))| \le \Delta$, it follows that $p_{\chi}(v) \le p_{\Delta}(v)$.

We now proceed to the proof of the items in the conclusion of the lemma. The second item is trivial, and the fourth item follows as in \cite[Lemma~2.2(d)]{BC20}. The only non-trivial part of the first item is the claim about the probability with which $|L'_\out(v)|=1$, which we will check at the end of the proof. 
We now verify the third item. 

\begin{itemize}
\item The expression on the left hand side of \cref{eq:disjoint-pchi} is the probability that a particular $c\in E_L(v)\setminus\chi(N(v))$ is chosen as $\chi'(v)$, since for this to happen, we must have chosen the second case of $\tsc{disjoint.gen}$ (which happens with probability $(k-Q-D)/(k-Q-D/2)$), then the second subcase of this (which happens independently with probability $1-q(v)$), and then chosen $c\in E_L(v)$ from a uniform sample (which happens independently with probability $1/E$), all before choosing the last line of $\tsc{disjoint.decode}$ (which happens independently with probability $p_\chi'(v)$).
\item The expression on the left hand side of \cref{eq:disjoint-pchi'} is the probability that a particular $c\in D_L(v)\setminus\chi(N(v))$ is chosen as $\chi'(v)$. 
\begin{itemize}
\item The first term comes from the case where we generate $L_\out'(v) = \{c,c'\} = N(w)$ for some $w\in N_L^\ast(v)$ (which happens with probability $(1-(k-Q-D)/(k-Q-D/2))\times 2/D = 1/(k-Q-D/2$)) -- note that this always decodes to $c$ since $c$ is the unique element in $\{c,c'\}$ which is not in $\chi(N(v))$). 
\item The second term is similar to the previous paragraph -- specifically, we must choose the second case of $\tsc{disjoint.gen}$ (which happens with probability $(k-Q-D)/(k-Q-D/2)$), then the first subcase of that (which happens independently with probability $q(v)$), choose $c \in D_L(v)$ from a uniform sample (which independently happens with probability $1/D$), enter line 12 of $\tsc{disjoint.gen}$ (which happens independently with probability $p_{\Delta}(v)$), and finally, enter line 30 of $\tsc{disjoint.decode}$ (which happens independently with probability $p_{\chi}(v)/p_{\Delta}(v)$).
\end{itemize}
\item Finally, all remaining colors in $[k]\setminus\chi(N(v))$ are in $[k]\setminus S_L(v)$, and are treated uniformly, hence as before we have the desired uniformity.
\end{itemize}


Finally, we verify the remaining claim in the first item. Indeed, the bounding chain gives a set of size $1$ with probability
\begin{align*}
\frac{k-Q-D}{k-Q-D/2}q(v)(1-p_\Delta(v)) &= \frac{k-Q-D}{k-Q-D/2}\bigg(1 - \frac{(k-Q-D/2)E}{(k-Q-D)(k-\Delta)}\\
&\qquad\qquad\qquad\qquad\qquad - \frac{(\Delta-Q-D/2)D}{(k-\Delta)(k-Q-D)}\bigg)\\
&= 1 - \frac{S-Q}{k-\Delta} + \frac{D/2}{k-Q-D/2},
\end{align*}
as desired. 
\end{proof}

\section{Analysis of \texorpdfstring{$\tsc{SamplerUnit}$}{SAMPLERUNIT}}
\label{sec:analysis}

Let $\tsc{SamplerUnit}$ be defined as in \cref{sec:alg}. More formally, let $T$ be the total number of updates used in the four phases, and starting from time $-T$, let $ (\alpha_{-T})_{\out},\dots,(\alpha_{-1})_{\out}$ be the $T$ updates described in \cref{sec:alg} generated as follows: for each $t \in [T]$, $(v_{-t})_{\init}$ and $(\gamma_{-t})_{\init}$ are chosen as described in \cref{sec:alg}. Moreover, 
\[((\tau_{-t})_{\init}, (L_{-t})_{\init}, (L'_{-t})_{\init}, (M_{-t}){_\init}) = ((\tau_{-t-1})_{\out}, (L_{-t-1})_{\out}, (L'_{-t-1})_{\out}, (M_{-t-1}){_\out} ),\]
with the initial conditions
\[((\tau_{-T})_{\init}, (L_{-T})_{\init}, (L'_{-T})_{\init}, (M_{-T}){_\init}) = (1, \prod_{v \in V(G)}[k], \prod_{v\in V(G)}[k], \emptyset).\]
We slightly overload notation (this does not create any confusion) by using $F$ to refer both to the sequence of tuples 
$(\alpha_{-T})_{\out},\dots,(\alpha_{-1})_{\out}$, as well as the composite function
$(\alpha_{-1})_{\out}.\tsc{decode}\circ \dots \circ (\alpha_{-T})_{\out}.\tsc{decode},$
interpreted in the obvious way, generated by these tuples. Also, the predicate $\Phi(F)$ is defined as  evaluating to $\tsc{true}$ iff $(L'_{-1})_{\out}$ is a list of sets of size $1$.\\

To complete the proof of \cref{thm:main}, we need to check two things: 
\begin{enumerate}[(Q1)]
\item Our description of $\tsc{SamplerUnit}$ is well-defined i.e.\ the promise required to execute $\tsc{seeding}$ and $\tsc{disjoint}$ is satisfied at every step. 
\item The resulting $\mathcal{D}, \Phi$ satisfy properties (P1)-(P4) of \cref{lemma:CFTP}. 
\end{enumerate}
The next subsection contains an analysis of Phase 4 of $\tsc{SamplerUnit}$, and completely addresses (Q2).

\subsection{Drift analysis: Phase 4 succeeds with probability at least \texorpdfstring{$1/2$}{1/2}}\label{sub:drift}
The goal of this subsection is to show that after all bounding sets have been reduced to size at most $2$, applying $\tsc{disjoint}$ a sufficient number of times at a uniformly randomly chosen vertex gives coalescence with sufficiently high probability. This is the analogue of \cite[Lemma~2.5]{BC20}. As in \cite{BC20,Hub98}, we will make use of the following random walk lemma due to Huber \cite{Hub98} (stated below with minor indexing errors corrected).
\begin{theorem}[{\cite[Theorem~4]{Hub98}}]\label{thm:rand-walk}
Suppose that $X_t$ is a random walk on $\{0,1,\ldots,n\}$ where $0$ is a reflecting state and $n$ is an absorbing state. Further, assume that $|X_{t+1}-X_{t}|\le 1$, and $\mb{E}[X_{t+1}-X_{t}~|~X_t=i]\ge \kappa_i>0$ for all $X_t<n$.
Let $e_i$ be the expected number of times the walk hits the state $i$. 
Then, \[\sum_{i=0}^{n-1}e_i\le \sum_{i=0}^{n-1}\frac{1}{\kappa_i}.\]
\end{theorem}

\begin{lemma}\label{lem:drift}
Assume $k > 5\Delta/2$ and let $T_{D} = 2\frac{k-\Delta}{k-5\Delta/2}n\log n$. Suppose that we have an update $\alpha_0 = \alpha$ with bounding list $L'$ satisfying $|L'(v)|\le 2$ for all $v\in V(G)$. Consider a random sequence of $\tsc{disjoint}$ updates $\alpha_1,\ldots,\alpha_{T_D}$ generated in sequence (note the forward time indexing) via
\[\alpha_t = \tsc{disjoint.gen}[\alpha_{t-1},\tsc{unif}(V(G))]\]
for all $1\le t\le T_{D}$. Let $\alpha_{T_D} = (\cdot,\cdot,\cdot,L_0',\cdot,\cdot)$. Then $|L_0'(v)| = 1$ for all $v\in V(G)$ with probability at least $1/2$.
\end{lemma}
\begin{remark}
Part of the assertion of this lemma is that the promise required to execute $\tsc{disjoint}$ is satisfied throughout. 
\end{remark}
\begin{proof}
Let $W_t = \{v\in V(G): |L_t'(v)| = 1\}$, let $X_t = |W_t|$, and let $\ol{W_t} = V(G)\setminus W_t$. Clearly, $X_t\in\{0,\ldots,n\}$ and $|X_{t+1}-X_t|\le 1$, since each update changes the bounding list for at most one vertex. We now show that
\begin{equation}\label{eq:drift}
\mb{E}[X_{t+1}-X_t|X_t]\ge\frac{n-X_t}{n}\bigg(1-\frac{3\Delta/2}{k-\Delta}\bigg).
\end{equation}
After proving \cref{eq:drift}, the claim follows immediately by \cref{thm:rand-walk}, since $n$ is easily verified to be an absorbing state (because in such a situation, any vertex $v\in V(G)$ has $|S_L(v)| - |Q_L(v)| = |D_L(v)| = 0$).

To prove \cref{eq:drift} we first show that if $|L(w)| \le 2$ for all $w \in V(G)$, then
\begin{equation}\label{eq:weights}
|S_L(v)|-|Q_L(v)|\le\frac{3}{2}|\ol{W_t}\cap N(v)| + \frac{|N_L^\ast(v)|}{2} \le \frac{3}{2}|\ol{W_t}\cap N(v)| +  \frac{|D_L(v)|}{4} 
\end{equation}
To see this, consider assigning weights to each color: a color in $S_L(v)\setminus Q_L(v)$, which appears in $m$ bounding lists, is assigned weight $1/m$, and any other color is assigned weight $0$. In particular, for any $w\in N_L^\ast(v)$, both elements of $L(w)$ are weight $1$ so that the sum of weights in $L(w)$ for $w \in N_L^\ast(v)$ is $2$. Also, if $|L(w)| = 1$, then sum of weights is $0$ by definition. Finally, in all other cases the sum of weights is at most $1 + 1/2 = 3/2$ (since $w \notin N_L^*(v)$ implies that at least one of the two colors must appear in at least two bounding lists). Now, \cref{eq:weights} follows by noting that the leftmost quantity is the sum of all the weights in all $L(w)$ for $w\in N(v)$ (counting colors multiple times), whereas the middle quantity is a trivial upper bound for this sum given the information above, and the rightmost inequality follows by noting that $|N_L^\ast(v)| = |D_L(v)|/2$. 

Equation \cref{eq:weights} shows that the condition needed to apply $\tsc{disjoint.gen}$ is satisfied at every step it is used. Specifically, we find that \[|S_L(v)|-|Q_L(v)|\le\frac{3}{2}\Delta+\frac{|D_L(v)|}{4}\le (k-\Delta)\frac{k-|Q_L(v)|}{k-|Q_L(v)|-|D_L(v)|/2},\]
where the second inequality holds since  $D_L(v)\in[0,2\Delta]$, $k\ge 2.5\Delta$, and $|Q_L(v)|\ge 0$.

Moreover, dividing \cref{eq:weights} by $k - \Delta$, we immediately deduce that
\begin{equation}\label{eq:profit}
\frac{|S_L(v)|-|Q_L(v)|}{k-\Delta} - \frac{|D_L(v)|/2}{k-|Q_L(v)|-|D_L(v)|/2}\\
\le\frac{3}{2}\frac{|\ol{W_t}\cap N(v)|}{k-\Delta},
\end{equation}
where for the second term on the left hand side, we have used that $k-|Q_L(v)|-|D_L(v)|/2\le k\le 2(k-\Delta)$.

Therefore,
\begin{align*}
\mb{E}[X_{t+1}-X_t|L_t'] &= \frac{1}{n}\bigg[\sum_{v\notin W_t}\bigg(1-\frac{|S_L(v)|-|Q_L(v)|}{k-\Delta}+\frac{|D_L(v)|/2}{k-|Q_L(v)|-|D_L(v)|/2}\bigg)\\
&\qquad- \sum_{v\in W_t}\bigg(\frac{|S_L(v)|-|Q_L(v)|}{k-\Delta}-\frac{|D_L(v)|/2}{k-|Q_L(v)|-|D_L(v)|/2}\bigg)\bigg]\\
&= \frac{1}{n}\bigg[|\ol{W_t}| - \sum_{v\in V(G)}\bigg(\frac{|S_L(v)|-|Q_L(v)|}{k-\Delta}-\frac{|D_L(v)|/2}{k-|Q_L(v)|-|D_L(v)|/2}\bigg)\bigg]\\
&\ge\frac{1}{n}\bigg[|\ol{W_t}|-\frac{3}{2}\sum_{v\in V(G)}\frac{|\ol{W_t}\cap N(v)|}{k-\Delta}\bigg]\\
&\ge\frac{|\ol{W_t}|}{n}\bigg[1-\frac{3\Delta/2}{k-\Delta}\bigg],
\end{align*}
where the penultimate inequality uses \cref{eq:profit} and the last inequality follows since the graph has maximum degree at most $\Delta$. Finally, the law of total expectation and $|\ol{W_t}| = n - X_t$ gives the desired inequality.
\end{proof}

In particular, \cref{lem:drift} shows that once $\tsc{SamplerUnit}$ reaches Phase 4, it succeeds with probability at least $1/2$. Therefore, it only remains to check that the first three phases are always completed successfully i.e. the various guarantees required by $\tsc{compress}, \tsc{seeding}$, and $\tsc{disjoint}$ are satisfied throughout the first three phases. 
\subsection{Phases \texorpdfstring{$1$}{1}, \texorpdfstring{$2$}{2}, and \texorpdfstring{$3$}{3} always succeed}
\begin{lemma}\label{lem:phase-1}
Phase $1$ succeeds deterministically, i.e., every application of $\tsc{compress}$ and  $\tsc{seeding}$ is guaranteed to satisfy its promise.
\end{lemma}
\begin{proof}
First, note that for every $i \in [s]$, we can indeed find the necessary associated set $A$ -- $v_i$ has at most $\Delta/3$ neighbors in $\mc{S}$ by construction, and each neighbor of $v_{i}$ in $\{v_{j}: j < i\}$ has a bounding list of size at most $3$ (since $\tsc{seeding}$ has already been applied to such a vertex); hence $A$ needs to contain a union of at most $\Delta/3$ sets of size at most $3$.

Next, note that when we perform $\tsc{seeding}$ on $v_i$, we trivially have $|S_L(v)|\le 2\Delta$ since each vertex to which we apply $\tsc{compress}$ (i.e.\ those neighbors of $v_i$ which do not precede it in $\mc{S}$) contributes up to $1$ additional color not present in $A$. The claim now follows upon noting that
\[\frac{|S_L(v)|^2}{\Delta+|S_L(v)|}\le\frac{(2\Delta)^2}{\Delta+(2\Delta)} = \frac{4\Delta}{3}\le k-\Delta.\qedhere\]
\end{proof}
\begin{lemma}\label{lem:phase-2}
Phase $2$ succeeds deterministically, i.e., every application of $\tsc{disjoint}$ is guaranteed to satisfy its promise.
\end{lemma}
\begin{proof}
The existence of the set $A$ follows as in the previous proof by noting that each vertex has at most $\Delta/3$ neighbors in $\mc{S}$, each of which has a bounding list of size at most $3$ at the end of Phase 1. 

Further, since each $v_i$ has at most $(1-\eta)\Delta$ neighbors outside $\mc{S}$, it follows that after applying $\tsc{compress}$ to all neighbors of $v_i$ not in $\mc{S}$, we have $|S_L(v)|\le \Delta+(1-\eta)\Delta = (2-\eta) \Delta$, where the first inequality follows since the $(1-\eta)\Delta$ neighbors of $v_i$ outside $\mc{S}$ can each contribute at most $1$ color outside of $A$ to $S_L(v)$. 
The claim now follows upon noting that
\[|S_L(v)|-|Q_L(v)|\le(2-\eta)\Delta < k-\Delta\le(k-\Delta)\bigg(\frac{k-|Q_L(v)|}{k-|Q_L(v)|-|D_L(v)|/2}\bigg).\qedhere\]
\end{proof}

The analysis of Phase 3 is more nontrivial due to the intricate nature of choosing the set $A$ of size $\Delta$. Ultimately the proof is a routine casework check. 
\begin{lemma}\label{lem:phase-3}
Phase $3$ succeeds deterministically, i.e., every application of $\tsc{disjoint}$ is guaranteed to satisfy its promise.
\end{lemma}
\begin{proof}
Let $v = v_i$ for some $s+1 \le i \le n$. 
As in the definition of Phase $3$, let $L_m$ be current the bounding list, restricted to marked neighbors of $v$. Also, let $L$ be the bounding list \emph{after} applying $\tsc{compress}$ to all unmarked neighbors of $v$. Note that $v$ has at most $(1-\eta)\Delta$ unmarked neighbors; suppose that it has $\eta'\Delta$ unmarked neighbors (this is true even for the neighbors of $v$ in $\mc{S}^{c}$). Recall that $\eta = 1/3 - 2\sqrt{(\log \Delta)/\Delta}$ for a sufficiently large constant $C$ and that $S_{L_m}(v)$ is the disjoint union $Q_{L_m}(v) \cup E_{L_m}(v) \cup D_{L_m}(v)$. Recall also that for all marked neighbors $w$ of $v_i$, $|L_{m}(w)|\le 2$.  \\

\textbf{Case 1: }$|S_{L_m}(v)| \le \Delta$. In this case, we must have that $S_{L_m}(v) \subset A$, since colors in $S_{L_m}(v)$ are chosen to be in $A$ before any other colors.  Then, as in the proof of \cref{lem:phase-2}, we see that $|S_L(v)| \le \Delta + (1-\eta)\Delta$, so that as before,
\[|S_L(v)| - |Q_L(v)|\le (2-\eta)\Delta< k - \Delta\le (k-\Delta)\frac{k-|Q_L(v)|}{k-|Q_L(v)|-|D_L(v)|/2}.\]

\textbf{Case 2: }$A\subseteq Q_{L_m}(v)\cup E_{L_m}(v)$.  
In particular, we must have $X = |Q_{L_m}(v)\cup E_{L_m}(v)|\ge\Delta$. Let $Y = |D_{L_m}(v)|$. Since the bounding list of each marked vertex has size at most $2$, we may use a similar argument as in the proof of \cref{eq:weights} to see that the total weight (as defined there) of the bounding list of each marked vertex $w$ intersecting $Q_{L_m}(v) \cup E_{L_m}(v)$ is at most $3/2$. Since $X$ is at least the sum of all the weights in all such lists (counting colors multiple times), it follows that there are at least $2X/3$ different (marked) $w\in N(v)$ with $L_m(w)$ intersecting $Q_{L_m}(v) \cup E_{L_m}(v)$.
Therefore there are at most $\Delta-\eta'\Delta-2X/3$ marked neighbors intersecting $D_{L_m}(v)$, so that 
\[|S_{L_m}(v)|\le X + 2(\Delta-\eta'\Delta - 2X/3)\le 5\Delta/3 - 2\eta'\Delta.\]
Finally, after applying  $\tsc{compress}$ to the unmarked neighbors of $v$, we gain an additional at most $\eta'\Delta$ elements. Thus $|S_L(v)|\le (5/3-\eta')\Delta$, and as above, the result follows immediately since $5/3\le 2-\eta$.\\

\textbf{Case 3: } $Q_{L_m}(v)\cup E_{L_m}(v)\subseteq A\subseteq S_{L_m}(v)$. Again let $X = |Q_{L_m}(v)\cup E_{L_m}(v)|$ and $Y = |D_{L_m}(v)|$. Thus $X\le\Delta\le X + Y$.

First, by repeating the computation in Case 2, but with the trivial lower bound $X \ge 0$, we see that $X+Y \le 2(1-\eta')\Delta$ and  $|S_L(v)|\le 2(1-\eta')\Delta+\eta'\Delta = (2-\eta')\Delta$. Thus if $\eta'\ge\eta$, we are done as before. Hence, we may assume that  $\eta'\in[0,\eta]$.

We ultimately want to check the condition
\begin{equation}
|S_L(v)|-|Q_L(v)| < (k-\Delta)\frac{k-|Q_L(v)|}{k-|Q_L(v)|-|D_L(v)|/2}.
\end{equation}
Since the left hand side is decreasing and the right hand side is increasing in $|Q_L(v)|$, it suffices to check $|S_L(v)| < (k-\Delta)k/(k-|D_L(v)|/2)$, i.e.
\begin{equation}\label{eq:quadratic-check}
|S_L(v)|\bigg(k-\frac{|D_L(v)|}{2}\bigg) < k(k-\Delta).
\end{equation}
Note that there are at least $\lfloor (X+Y-\Delta)/2\rfloor$ pairs of colors $L_m(w)$, for $w\in N_{L_m}^\ast(v)$, inside $D_{L_m}(v)\setminus A$.
Note also that every additional color coming from the $\eta' \Delta$ unmarked neighbors could be one of the following: (i) a color outside of $S_{L_m}(v)$ (ii) a color in $D_{L_m}(v)$ (observe that every such color prevents two colors from $D_{L_m}(v)$ from appearing in $D_L(v)$), and 
(iii) a color in $A \setminus D_{L_m}(v)$. Suppose we have $s\Delta$ colors of type (i) and $t\Delta$ colors of type (ii). Then,   
\[|S_L(v)| = X + Y + s\Delta,\quad|D_L(v)|\ge 2\lfloor(X+Y-\Delta)/2\rfloor-2t\Delta,\quad s + t \le \eta'.\]
Observe that if $|S_L(v)| = X + Y + s\Delta \le k-\Delta$, then \cref{eq:quadratic-check} is trivially satisfied. Therefore, we may assume that $X + Y \ge k - \Delta - s\Delta \ge (2-\eta - s)\Delta$. 
Finally, let $X = x\Delta$, $Y = y\Delta$, and $k = \kappa\Delta$, so that \cref{eq:quadratic-check} follows if
\begin{equation}
\label{eq:optimization}
(x+y+s)\bigg(\kappa+t-\frac{x+y-1-\Delta^{-1}}{2}\bigg) < \kappa(\kappa-1).
\end{equation}
From the discussion above, we have the constraints $z = x+y\in[2-\eta-s,2-2\eta']$, $s,t\ge 0$, $s+t\le \eta'$, $\eta'\in[0,\eta]$, and  $\kappa\ge 3-\eta$. Recall also that $\eta$ is a fixed constant less than $1/3$. 
Also, increasing $S_L(v) \setminus D_L(v)$ can only make \cref{eq:quadratic-check} harder to satisfy, we may assume that $s + t = \eta'$. 

We see by taking derivatives that as long as $\kappa\ge 3/2$ (which is true in our case), the condition in \cref{eq:optimization} is most restrictive when $\kappa$ is taken smaller. Therefore, we may let $\kappa = 3-\eta$, to see that \cref{eq:optimization} is implied by 
\begin{equation}
\label{eq:optimization-1}
(z+s)\bigg(3-\eta+t-\frac{z-1-\Delta^{-1}}{2}\bigg) < (2-\eta)(3-\eta).
\end{equation}
Next, we see by taking derivatives that for $\eta\in[0,1/3]$, the condition in \cref{eq:optimization-1} is strictly more restrictive when $\eta = 1/3$. Therefore, by taking $\eta = 1/3$ (note that for us, $\eta$ is strictly smaller than $1/3$), we see that \cref{eq:optimization-1} is implied by 
\begin{equation}
\label{eq:optimization-2}
(z+s)\bigg(\frac{8}{3}+t-\frac{z-1-\Delta^{-1}}{2}\bigg)\le\frac{40}{9}.
\end{equation}
At this point, we assume $\Delta\ge 9$, and reduce to checking
\begin{equation}
\label{eq:optimization-3}
(z+s)\bigg(\frac{29}{9}+t-\frac{z}{2}\bigg)\le\frac{40}{9}
\end{equation}
on the region carved out by $s,t\ge 0$, $s + t\le 1/3$, $z\ge 5/3-s$, and $z\le 2-2s-2t$.

Note that the left hand side is a downward quadratic in $z$ with maximum at $z = 20/9+t-s/2$, which is always bigger than $2-2s-2t$. Hence, the left hand side is maximized at $2-2s-2t$, which yields
\[(2-s-2t)\bigg(\frac{20}{9}+s+2t\bigg)\stackrel{?}{\le}\frac{40}{9};\]
this is clearly true since $s+2t\ge 0$.
\end{proof}

\subsection{Putting everything together}
We now quickly check that (Q1) and (Q2) follow from our work so far. Indeed, \cref{lem:phase-1,lem:phase-2,lem:phase-3,lem:drift} shows that (Q1) is true. For (Q2), we note that (P2) follows from (C2), (S2) and (D2), and that (P1) follows from (C3), (S3), and (D3). Moreover, the same running time analysis as in \cite{BC20} shows that (C4), (S4) and (D4) easily imply (P4). Finally, \cref{lem:drift} implies (P3), which completes our analysis.

\section{Conclusion and Open Problems}\label{sec:conclusion}
We first briefly elaborate on how the above analysis can be extended to push slightly beyond $k \geq (8/3 + o(1))\Delta$, i.e., to perfectly sample $(8/3-\epsilon)\Delta$ colors for some absolute constant $\epsilon \approx 10^{-2}$. The current algorithm (for sufficiently large $\Delta$) is only limited at $k = (8/3 + o(1))\Delta$ in Phase $2$ (although this requires performing the analysis in \cref{lem:phase-3} more carefully). In order to improve Phase $2$, a variant of \cref{alg:disjoint} which allows for disjoint triples works. (There is an even more efficient routine using both disjoint pairs and disjoint triples.)

However, all of these techniques are currently limited at $k > 5\Delta/2$ due to \cref{lem:drift}; the extremal configuration limiting the analysis here has shadows of the configurations which limit Jerrum's \cite{Jer95} analysis for approximate sampling at $k > 2\Delta$ using the Glauber dynamics. However, incorporating the techniques of Vigoda \cite{Vig00} and more general path-coupling ideas \cite{BD97} may allow one to break this barrier, but the interaction of these techniques with the bounding chain framework of \cite{Hag98,Hub98} is nontrivial and remains an interesting open question.

\bibliographystyle{amsplain0.bst}
\bibliography{main.bib}

\appendix

\section{Proof of \texorpdfstring{\cref{prop:seeding}}{Proposition 3.1}}
\label{sec:proof-find-seeding-set}
\begin{proof}
The proof uses the symmetric Lov\'asz Local Lemma (LLL; see \cite{AS16}). We sample $\mc{S}$ by selecting each vertex in $V(G)$ independently with probability $(\eta + 1/3)/2$. For each $v \in V(G)$, let $\mc{B}_v$ denote the `bad' event that $v$ does not satisfy the condition in \cref{eq:seeding-conclusion}. Then, it is straightforward to verify that each bad event is mutually independent from all but at most $d = \Delta^2$ other bad events (corresponding to vertices with distance at most $2$ from $v$). Moreover, a standard application of the Chernoff bound shows that each bad event occurs with probability at most $p = \exp(-\Omega_\eta(\Delta))$. Thus, for $\Delta\ge C_\eta$, we trivially have $ep(d+1) < 1$, which guarantees that a set $S$ satisfying \cref{eq:seeding-conclusion} for all $v \in V(G)$ exists.

In order to algorithmically generate such a set, we use the algorithmic version of LLL due to Moser and Tardos \cite{MT10}. We set $x(v) = ep < 1/(d+1)$ for all $v\in V(G)$, which can easily be checked to satisfy the hypotheses of \cite[Theorem~1.2]{MT10}. Therefore, by  \cite[Theorem~1.2]{MT10}, the expected number of `resampling operations' is at most
\[\sum_{v\in V(G)}\frac{x(v)}{1-x(v)} < 2\sum_{v\in V(G)}x(v) = O(n\exp(-\Omega_\eta(\Delta))).\]

For an analysis of the running time, note that it takes time $O(n\Delta)$ to initially sample and compute the number of neighbors in $\mathcal{S}, \mathcal{S}^{c}$ each vertex has, which we maintain as an array throughout. We also maintain a binary heap with the set of vertices violating \cref{eq:seeding-conclusion}. Since each resampling operation amounts to resampling the neighbors of a violating vertex, we see that, in particular, each resampling operation requires updating at most $\Delta^{2}$ array elements (corresponding to vertices within distance $2$ of the violating vertex at which resampling occurs). Finally, we remove vertices which no longer violate \cref{eq:seeding-conclusion} from our binary heap, and add vertices which have turned into violators to the binary heap, which takes time $O(\Delta^2\log n)$. Therefore the running time is $O(n\Delta + n(\log n)\Delta^2\exp(-\Omega_\eta(\Delta)))$, which is $O(n\Delta + n\log n)$ as desired. Note that we terminate early if the number of resampling operations is twice the expectation in order to obtain the failure probability (by Markov's inequality) and running time guarantee in the statement of the proposition. 
\end{proof}

\end{document}